\documentclass[10pt, journal]{IEEEtran}

\usepackage{epsfig}
\usepackage{amssymb} 
\usepackage[tbtags]{amsmath} 
\usepackage{graphics,eepic,epic}
\usepackage{latexsym}
\usepackage{euscript}
\usepackage{subfigure}
\usepackage{graphics,eepic,epic,psfrag}
\usepackage{algorithm, algorithmic}

\newcommand{\cR}{\mathcal{R}}
\newcommand{\cP}{\mathcal{P}}

\newcommand{\cT}{\mathcal{T}}

\newcommand{\bx}{\mathbf{x}}

\newtheorem{theorem}{Theorem}

\newtheorem{definition}{Definition}
\newtheorem{lemma}[theorem]{Lemma}

\usepackage{cite}

\begin{document}

\title{\begin{huge}
Cooperative Packet Routing using Mutual Information Accumulation
\end{huge}}
\pagenumbering{arabic}


\author{Yanpei Liu,~\IEEEmembership{Student Member,~IEEE,}
	 Jing Yang,~\IEEEmembership{Member,~IEEE,}
  	Stark C. Draper,~\IEEEmembership{Member,~IEEE} \thanks{This
    work was presented in part at the 49th annual Allerton Conference on
	Communication, Control, and Computing, Monticello IL, September
    2011.}  \thanks{The authors are with the Department of Electrical and
    Computer Engineering, University of Wisconsin, Madison, WI 53706
    (E-mail: \{yliu73@, yangjing@ece.,
    sdraper@ece.\}wisc.edu).}  \thanks{The authors were supported by the National Science Foundation under grant CCF-0963834, Air Force Office of Scientific Research under grant FA9550-09-1-0140 and by a grant from the Wisconsin Alumni Research Foundation.}}

\maketitle
\pagenumbering{arabic}

\begin{abstract}
We consider the resource allocation problem in cooperative
wireless networks wherein nodes perform mutual information
accumulation.  We consider a unicast setting and arbitrary arrival
processes at the source node.  Source arrivals can be broken down into
numerous packets to better exploit the spatial and temporal diversity
of the routes available in the network.  We devise a
linear-program-based algorithm which allocates network resource to meet a certain transmission objective.  
Given a network, a source with multiple arriving
packets and a destination, our algorithm generates a policy that
regulates which nodes should participate in transmitting which
packets, when and with what resource. By routing different packets
through different nodes the policy exploits spatial route diversity,
and by sequencing packet transmissions along the same route it exploits
temporal route diversity.
\end{abstract}


\section{Introduction}

In multi-hop cooperative relay networks, relay nodes cooperate with
each other to deliver data from network source to network sink. The
benefits of cooperation include improvements in energy efficiency, in
robustness to fading and interference, and reductions in the likelihood
of the loss of connectivity \cite{Draper_2008, ExOR}.  Cooperation arises naturally in
the wireless medium due to its broadcast nature.  When two nodes are
communicating, neighboring nodes can listen in on the transmission and
thus are in a position to help.  As an example, nodes further from the
source and closer to the destination can overhear earlier nodes'
transmissions.  This gives them a head start on decoding the message
and speeds up message delivery in contrast to conventional multi-hop
transmission.

In order to be able to exploit the broadcast nature of wireless
transmission, we assume nodes are equipped with the following
abilities:
\begin{itemize}
\item Nodes can listen in on all transmissions.
\item Nodes store all observations and slowly accumulate noisy observations
until they can decode the message.
\end{itemize} 
To simplify our model we assume transmitters operate on orthogonal
channels and thus there is no interference.  The first point
simply means nodes continually monitor the full bandwidth of the
system.  The second point means that nodes have large memories and
thus can store all observations until they become useful.  This is
what we mean by mutual information accumulation.   Mutual information accumulation can be realized through the use of rateless (or ``fountain'') codes \cite{Draper_2008, Castura_2007, Gummadi_2008, Mitzenmacher_2004, LT_Code, Raptor_Code}.

Enabling nodes with this physical layer (PHY) ability changes the routing and resource allocation problem from that traditionally considered (e.g., in backpressure \cite{Tassiulas_1992, Tassiulas_1993, Neely_2008} or in network coding \cite{COPE}). Simpler PHY layer modiﬁcations such as energy accumulation using
space-time or repetition coding have been considered in \cite{Maric_2004, Maric_2005, Chen_2005}. The difference between energy accumulation and mutual information accumulation can be easily understood from the following example \cite{Draper_2008}. Consider binary signaling over a pair of independent erasure channels each having erasure probability $p_e$ from two relays to a single receiver. If the two relays use repetition coding, corresponding to energy accumulation, then each symbol will be erased with probability $p_e^2$. Therefore, $1-p_e^2$ symbols are successfully received on average per transmission. Instead, if we use different codes, the transmissions are independent and on average $2(1-p_e)$ novel parity symbols are received per transmission. 


In contrast to previous works that consider the routing of a single
packet \cite{Draper_2008, Urgaonkar_2010}, in this work we consider general arrival processes
and the routing of multiple packets.  This generalization is important
to consider.  With multiple packets on the go there are extra spatial
and temporal degrees of freedom to consider that do not arise in the
single packet case. As we show in this paper, by using our algorithm designed for multiple
packets transmission, we can exploit the route diversity of the
network in both spatial and temporal aspects. Different packets may be routed along partially (or wholly) disjoint
routes, thereby exploiting spatial or route diversity. This is
a desired feature when packets compete for resource; it may be
desirable to route a group of packets through one link and the rest
through another to balance the traffic and resource consumption. Depending on
the arrival process different packets may be sent in sequence along
the same route, thereby exploiting temporal diversity. This usually happens when routes choices are
limited or packets arrive at different times. Essentially, our algorithm exploits available resource within the network and
distributes it efficiently.  

The main contribution of our work is the design of a joint routing and resource allocation algorithm for multiple packets transmission from the source to the destination in a given network. Our work is a natural extension of the work in \cite{Draper_2008} where only single packet transmission is considered. The contribution of our work is twofold. We first present a formulation of the routing and resource allocation problem. The formulation considers various forms of energy and bandwidth constraints and is expressed in the form of linear program (LP). Second, we design a centralized algorithm which solves the routing and resource allocation \textit{jointly} by solving a sequence of LPs. Each LP solves for the optimal resource allocation given a route decision of all packets. The resource allocation result is then used to update the route decision and the method proceeds iteratively.

Our work differs from previous cooperative routing and resource allocation works in the sense that we consider \textit{both} PHY layer and routing (rather than just the PHY layer in simple networks or both in simplified settings). It differs from the current state-of-the-art cooperative routing algorithms (e.g., opportunistic multi-hop routing \cite{ExOR}, backpressure-based algorithm \cite{Neely_2008, Tassiulas_1992, Tassiulas_1993}, and network-coding-based algorithm \cite{MIXIT, Tracey_2006}) in that it involves a modified PHY layer with an extra feature -- mutual information accumulation which can be easily implemented using rateless codes. Meanwhile, there are many existing works in cooperative routing using this feature. In \cite{Castura_2007}, mutual information is considered for single relay networks. In \cite{Molisch_2007}, mutual information is also considered, but without a consideration on resource allocation. Routing and resource allocation with mutual information accumulation is considered in \cite{Draper_2008, Urgaonkar_2010}, but they only consider single packet routing. To the best of our knowledge, there has been little prior work investigating routing and resource allocation with mutual information accumulation for multiple packets.

The rest of the paper is organized as follows. We present the system model and problem formulation in Sec.~\ref{sec.system_model}. The centralized algorithm is developed and discussed in Sec.~\ref{sec.centralized_algorithm}. We provide detailed numerical results in Sec.~\ref{sec.simulation_results}. Discussion and conclusion are in Sec.~\ref{sec.discussion_and_conclusion}.
\section{System Model}
\label{sec.system_model}
We consider a system with $L$ nodes: the source, always labeled $1$, the destination, always labeled $L$ and $L-2$ relay nodes. Suppose that $K$ large data files (not necessarily of the same size) arrive at the source at a set of specified times $\tau_1, \tau_2, \ldots, \tau_K$. Upon arrival each file is subdivided into a number of packets.  In all there are $N \geq K$ packets.  The indices of the packets are ordered according to the arrival times of the corresponding data files. We consider each of these packets as a separate commodity from a routing point of view.
The network's goal is to deliver $N$ such packets from the source to the destination under various channel constraints. Each relay may participate in transmitting a subset of these packets or remain silent. To simplify the analysis we assume that the only significant power expenditure for each node lies in transmission.

\subsection{Problem parameterization}
We first introduce the definition of a ``decoding event''.

\begin{definition}
A decoding event occurs when either a node decodes a packet, or a packet is made available at the source.
\end{definition}

We assume a decode-and-forward relaying strategy. We denote as $T_i^c$ the time at which node $i$ decodes packet $c$, $c \in \{1, 2, \ldots, N\}$. In particular, each $T_1^c$ corresponds to the arrival time of the corresponding data file $\ell$ which packet $c$ is subdivided from, i.e., $T_1^c = \tau_{\ell}$. We next present the definition of a ``decoding order'' -- a key component in our algorithm. 

\begin{definition}
A decoding order $\cT$, containing $T_{1}^c$ and $T_{L}^c$ for all $c \in \{1, 2, \ldots, N\}$, is a subset of the set $\{T_1^1, \ldots, T_1^N, \ldots, T_L^N \}$ sorted in increasing order. 
\end{definition}

It should be noted that a decoding order $\cT$ only specifies the order of its decoding events. It does not specify when those events occur (or the value of $T_i^c$'s). 

We let $\cT$ denote the decoding order and let $M = |\cT|$ denote the number of decoding events inside this decoding order $\cT$. Since a node cannot transmit a packet until it has decoded it, the order of decoding events puts constraints on the resource allocated to nodes for each packet. We further note that a decoding order can contain up to $NL$ decoding events, i.e., $M \leq NL$. Then the label pair $(i, c)$ of $T_i^c$ uniquely determines its position $s$, $s \in \{1, 2, \ldots, M\}$. In other words, there exists a unique mapping $f$ which maps $(i, c)$ to $s$ and we further denote this mapping as $T_i^c \equiv T_{f(i, c)} \equiv T_s$. 

With a decoding order $\cT$, we define the ``inter-decoding intervals'' for its decoding events. 
\begin{definition}
Given a $\cT$, the inter-decoding interval $\Delta_s$ is the time interval between $T_s$, $T_{s-1} \in \cT$ and is defined as
\begin{align*}
\Delta_s = T_s - T_{s-1},
\end{align*}
where we denote time $0$ as $T_0$.
\end{definition} 
The entire message transmission can be thought of as consisting $M$ intervals. The $s$th interval is of duration $\Delta_s$ and is characterized by the fact that at the end of this $\Delta_s$, a decoding event associated with $T_s$ \textit{must} occur. A typical decoding order is shown in Fig.~\ref{fig.do}.
\begin{figure}[!t] 
  \centerline{\epsfig{figure=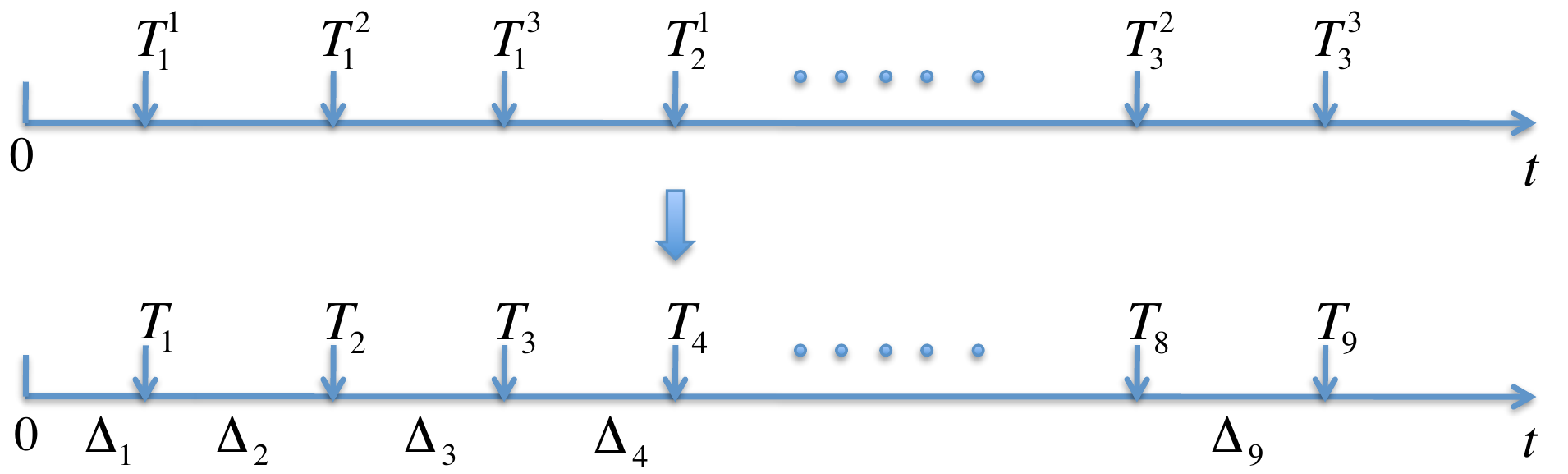,width=9cm}}
  \caption{A sample decoding order for $N = 3$, $L = 3$ and $M = NL$.}
  \label{fig.do}
\end{figure} 

We assume the $i$th node operates at a fixed power spectral density (PSD) denoted as $P_i$ (joules/sec/Hz), uniform across its transmission band. We assume the channel between any pair of nodes is block-fading and frequency non-selective. The channel gain between node $i$ and node $j$ is denoted as $h_{i,j}$. Under these assumptions, the spectral efficiency (bits/sec/Hz) between node $i$ and node $j$ can be expressed as \cite{Shannon_1948, Draper_2008}
\begin{align}
\label{eq.spectral_efficiency}
C_{i,j} = \log_2 \bigg [ 1+ \frac{h_{i,j} P_i W_i}{N_0 W_i} \bigg ] = \log_2 \bigg [ 1+ \frac{h_{i,j} P_i}{N_0} \bigg ], 
\end{align}
where $N_0/2$ denotes the PSD of the white noise process. We further denote the time-bandwidth product allocated to node $i$ to transmit packet $c$ during a given time interval $\Delta_s$ as $A_{i,s}^c$ (sec-Hz). Then the information flow from node $i$ to node $j$ during this interval is $A_{i, s}^c C_{i, j}$ bits.

\subsection{Problem constraints}
For a given decoding order we find the resource allocation minimizing an objective function subject to the following constraints:

\begin{enumerate}
\item $\Delta_s \geq 0$ for all $s$. 
\item $A_{i, s}^c \geq 0$ for all $s$, $i$ and $c$. 
\item Arrival process constraint; the source cannot transmit packet $c$ until its corresponding data file has reached the source. 
\item Decoding constraint; node $j$ must decode packet $c$ at the associated timing point $T_s$.
\item Constraint(s) on energy and bandwidth. 
\end{enumerate}

Arrival process constraint puts a constraint on value of $T_1^c$'s (thus the time intervals between them). Each $T_1^c$ corresponds to the arrival time of the corresponding data file $\ell$ which packet $c$ is subdivided from, i.e., $T_1^c = \tau_{\ell}$. We describe constraints $4)$ and $5)$ in more details below. 

We assume that nodes use codes that are ideal in the sense that they fully capture this potential flow, working at the Shannon limit at any rate. We also assume that distinct packets can be simultaneously transmitted from any node with no interference between packets and receivers have perfect knowledge to distinguish them. Nodes are further designed to use \textit{independently} generated codes. This design leads to another assumption that a receiver can combine information flows from two or more transmitters for each packet without any rate loss. As noted in \cite{Draper_2008}, the use of independently-generated codes is crucial for the mutual information accumulation process. If the transmitters used the \textit{same} code, the receiver would get multiple looks at each codeword symbol. This is ``energy accumulation.'' By looking at different codes the receiver accumulates mutual information rather than energy. 

With mutual information accumulation, the decoding constraint imposes a constraint on some timing points. If $T_j^c$, $j \neq 1$ is in the decoding order, node $j$ must be able to decode packet $c$ at this time. This constraint is formally expressed as
\begin{align}
\label{eq.decoding_con}
\sum_{i : f(i, c) < f(j, c)} \sum_{s = f(i, c) + 1}^{f(j, c)} A_{i, s}^c C_{i, j} \geq B^c, \mbox{ for all } c, j \neq 1,
\end{align}
where $B^c$ is the size of packet $c$. Recall that $C_{i, j}$ is the spectral efficiency (bits/sec/Hz) of the channel connecting node $i$ to node $j$. Equation (\ref{eq.decoding_con}) says that in order for node $j$ to decode packet $c$, the total accumulated information at node $j$ for packet $c$ must exceed $B^c$ bits by the time $T_j^c$. The non-ideal nature of existing implementations of rateless codes can be handled by incorporating an overhead factor $(1+\epsilon)$ into the right-hand side of (\ref{eq.decoding_con}). Also, the relation (\ref{eq.decoding_con}) expresses the constraint that node $i$ can transmit packet c to node $j$ only after it has decoded packet $c$.

Constraints on energy and bandwidth can either be system-wide constraints or be imposed on a node by node basis. We state various possible constraints in the followings. 

\begin{enumerate}
\item \textit{Per-node bandwidth constraint}: If node $i$ is assigned with bandwidth $W_i$, its resource allocated during given $\Delta_s$ must satisfy the following
\begin{align}
\label{eq.per-node_bw_con}
\sum_{c=1}^N A_{i, s}^c \leq \Delta_s W_i, \mbox{ for all } i, s. 
\end{align}
\item \textit{Sum bandwidth constraint}: If the total bandwidth $W_{T}$ is allocated, a sum bandwidth constraint applies across all nodes during any given $\Delta_s$. This can be formally expressed as
\begin{align}
\label{eq.sum_bw_con}
\sum_{i = 1}^L \sum_{c=1}^N A_{i, s}^c \leq \Delta_s W_{T}, \mbox{ for all } s. 
\end{align}
\item \textit{Per-node energy constraint}: If node $i$ with transmission power $P_i$ is assigned with energy budget $E_i$, we can express the per-node energy constraint as
\begin{align}
\label{per-node_eng_con}
\sum_{c=1}^N \sum_{s = f(i, c) + 1}^M A_{i, s}^c P_i \leq E_i, \mbox{ for all } i. 
\end{align}
\item \textit{Sum energy constraint}: If a sum energy $E_{T}$ is assigned for all nodes, a sum energy constraint can be applied as
\begin{align}
\label{eq.sum_eng_con}
\sum_{i = 1}^L \sum_{c=1}^N \sum_{s = f(i, c) + 1}^M A_{i, s}^c P_i \leq E_{T}. 
\end{align}
\end{enumerate} 

In the next section we develop a centralized algorithm based on the LP framework. 

\section{Centralized Algorithm}
\label{sec.centralized_algorithm}

The LP framework can take many objective functions. If there is only one data file (thus $T_1^1, T_1^2, \ldots , T_1^N = \tau_1$), one natural objective function to use is the total transmission time
\begin{align}
\label{eq.total_tx_time}
T_t = \sum_{s=1}^M \Delta_s.
\end{align}
Alternatively, if $K$ data files arrive at different times (thus packets are not available to the source at the same time), instead of minimizing total transmission time, one might want to minimize the average transmission time given by
\begin{align}
\label{eq.avg_tx_time}
T_a = \frac{1}{K} \sum_{\ell=1}^K T_{\ell}, 
\end{align}
where $T_{\ell}$ is the duration data file $\ell$ is being transmitted. This objective function is appropriate in the sense that it does not penalize the arrival time.
 
Other linear programming frameworks are also possible. For example, one may wish to minimize the total energy expenditure given by
\begin{align}
\label{eq.total_energy_exp}
\sum_{i = 1}^L \sum_{c=1}^N \sum_{s = f(i, c) + 1}^M A_{i, s}^c P_i,
\end{align}
subject to total transmission time or average transmission time constraint. 

\subsection{Characteristics of the problem}
We now study the properties of routing and resource allocation under different constraints with total transmission time (\ref{eq.total_tx_time}) as our objective function. First, under the per-node bandwidth constraint, we have the following lemma. 
\begin{lemma}
Under a given decoding order, suppose $\Delta_s$ and $A_{i, s}^c$ are the solution to the problem under the per-node bandwidth constraint (\ref{eq.per-node_bw_con}). Then for each $s$, there exists a node $i$ such that the inequality in (\ref{eq.per-node_bw_con}) becomes equality. 
\end{lemma}
\begin{proof}
We prove by contradiction. Suppose for some $s$ inequality is strict for all $i$, i.e., 
\begin{align*}
\sum_{c=1}^N A_{i, s}^c < \Delta_s W_i, \mbox{ for all } i. 
\end{align*}
Then we can scale down $\Delta_s$, yielding a smaller objective value. This contradicts the assumption that $\Delta_s$ is the optimal solution to the LP. 
\end{proof}

This proposition suggests that the $\Delta_s$ can be calculated as 
\begin{align*}
\Delta_s = \max_{i} \frac{\sum_{c=1}^N A_{i, s}^c}{W_i}.
\end{align*} 

Under the sum bandwidth constraint, we first have the following lemma. 
\begin{lemma}
\label{tm.sum_bw_lemma1}
Under a given decoding order, suppose $\Delta_s$ and $A_{i, s}^c$ are the solution to the problem under the sum bandwidth constraint (\ref{eq.sum_bw_con}). Then the inequality in (\ref{eq.sum_bw_con}) must be equality for all $s$. 
\end{lemma}
\begin{proof}
We prove this by contradiction. Suppose for some $s$ the equality does not hold, i.e., 
\begin{align*}
\sum_{i = 1}^L \sum_{c=1}^N A_{i, s}^c = \Delta_s W_s < \Delta_s W_T, \mbox{ for some } s. 
\end{align*}
Then for these $s$ we can scale down the corresponding $\Delta_s$ by $\frac{W_s}{W_T}$ while increase $W_s$ to $W_T$. We therefore obtain a solution which has smaller objective value. 
\end{proof}
 
This lemma leads to the following important theorem. 

\begin{theorem}
\label{tm.sum_bw_theorem}
Under the sum bandwidth constraint and a given decoding order, if $P_i = P$ for all $i$ then the solution that minimizes the objective in (\ref{eq.total_tx_time}) also minimizes the sum energy. 
\end{theorem}
\begin{proof}
The total energy expenditure of the entire system is
\begin{align*}
& \sum_{i = 1}^L \sum_{c=1}^N \sum_{s = f(i, c) + 1}^M A_{i, s}^c P \\
&= \sum_{i = 1}^L \sum_{c=1}^N \sum_{s = 1}^M A_{i, s}^c P \\
&= \sum_{s = 1}^M \Delta_s W_T P \\
&= T_tW_T P.
\end{align*}
The first equality holds because a node cannot transmit a given packet (resource allocated to it is zero) before it decodes the packet. The second equality follows from Lemma~\ref{tm.sum_bw_lemma1}. Since the objective $T_t$ is proportional to the energy used, minimizing one minimizes the other. 
\end{proof}

This theorem tells us that in this setting there is no trade off between total transmission time and energy. The minimum transmission time policy is identical to the minimum energy policy.

\begin{theorem}
\label{tm.sum_bw_theorem2}
Under the sum bandwidth constraint and a given decoding order, if the minimum transmission time for transmitting one packet is $T$, then the minimum transmission time for routing $N$ packets of same size is $NT$. 
\end{theorem}
\begin{proof}
Let $\cR$ denote the optimal set of routes on which transmitting a packet takes $T$ amount of time. These routes are equivalent in the sense that transmitting a packet on any of them takes $T$ amount of time, consuming $TW_T$ a amount of total resource. Then transmitting $N$ packets of the same size on $\cR$ can be finished within $NT$ amount of time by allocating each packet $\frac{TW_T}{N}$ amount of total resource. 

We now show that this is the best the system can do. Suppose otherwise, i.e., not all packets are routed through $\cR$ and the total transmission time is smaller than $NT$. Then there exists at least one packet traveling on a different set of equivalent routes $\cR' \neq \cR$. Since the transmission time is smaller than $NT$, each packets traveling on $\cR$ should be given more than $\frac{TW_T}{N}$ amount of resource. This implies that the packet on $\cR'$ shares less than $\frac{TW_T}{N}$ amount of resource whereas achieves a total transmission time less than $NT$. This contradicts our assumption that $\cR$ is the optimal set of equivalent routes. 
\end{proof}

This theorem suggests that in sum bandwidth constraint, the performance may not be improved by dividing a big data file into smaller packets. We will illustrate this in simulation.

\subsection{Optimizing the Decoding Order}
In this section we show an important theorem that motivates the centralized algorithm. The theorem tells us how to manipulate the decoding order based on the solution of the LP problem. Denote as $\cP$ a LP framework with a chosen objective and constraints being considered. Given any decoding order of length $M$, define
\begin{align*}
\bx = \bigg [ \Delta_1, \ldots, \Delta_M, A_{1, 1}^1, \ldots, A_{L, M}^N \bigg ]
\end{align*}
to be the solution to the LP with optimal objective value $T_{opt}$ on this decoding order. We then have the following theorem. 
\begin{theorem}
\label{tm.swap_tm}
If $\Delta_{m} = 0$ for some $m$ and we swap the positions of $T_{m}$ and $T_{m-1}$ in the decoding order, then the objective value $T_{opt}^*$ obtained with this swapped decoding order satisfies $T_{opt}^* \leq T_{opt}$.
\end{theorem}
\begin{proof}
We prove the theorem by showing that when swapped decoding order is used, the original solution $\bx$ with optimal objective value $T_{opt}$ is still feasible under new decoding order. To show $\bx$ is feasible under the new decoding order, we show the decoding constraint of the swapped decoding order is satisfied by this solution.

If a decoding event associated with $T_{m-1}$ has its decoding constraint satisfied at this time, its decoding constraint is certainly satisfied at a later time $T_{m}$. If a decoding event associated with $T_m$ has its decoding constraint satisfied at this time, it actually has the constraint satisfied at an earlier time $T_{m-1}$ since $A_{i, m}^c = 0$ for all $i$ and $c$. Therefore, solution $\bx$ with optimal objective value $T_{opt}$ is feasible under swapped decoding order
\end{proof}

The idea behind Theorem~\ref{tm.swap_tm} is based on the following observation. A solution to the LP with $\Delta_m = 0$ indicates that the event associated with $T_m$ takes place at exactly the same time with the previous event $T_{m-1}$, or actually occurs before it. Therefore, swapping the position of $T_m$ and $T_{m-1}$ typically gives a decrease in the objective value once the LP is solved under swapped decoding order. In the case that $T_{i}^c$ is swapped with $T_{L}^c$ for $i$ not equal to $L$, the decoding event $T_{i}^c$ is excluded from the new decoding order by our algorithm.   
\subsection{Centralized algorithm}
We design and implement a centralized resource allocation and routing algorithm iterating between two sub-problems:
\begin{enumerate}
\item For a given decoding order, a resource allocation scheme is determined by solving a given LP problem $\cP$. Solving the optimization problem with these constraints gives time allocation $\Delta_s$ (thus $T_s$) for all $s$ and the resource allocation $A_{i, s}^c$ for all $i$, $c$ and $s$. 
\item With a given resource allocation scheme, we update the decoding order. 
\end{enumerate} 
In the following we discuss the algorithm in more details. 
 
 \subsubsection{Initialization}
We first consider the decoding order initialization as if there were only one packet to be routed. We design a heuristic algorithm similar to the greedy filling algorithm in \cite{Maric_2004} based on network state. For single packet transmission, the decoding order is essentially the order of nodes which successfully decode that packet and can come online to broadcast it. Therefore, with a set of already decoded nodes, the next node to decode is \textit{most likely} the one which is able to benefit the most from the already decoded nodes. We thus choose the next node to be the one that has maximum sum spectral efficiency between itself and all previously decoded nodes. The algorithm pseudocode is given in Algorithm~\ref{alg.DOInit}. 

\begin{algorithm}
\caption{Initialize the decoding order}
\label{alg.DOInit}
\begin{algorithmic}[1]
\STATE $\cT \leftarrow \{T_1^1\}$
\WHILE {$length(\cT) < L-1$}
	\STATE $j \leftarrow \arg \max_{k \notin \cT, k \neq L} \sum_{i \in \cT} C_{i, k}$
	\STATE $\cT \leftarrow \{\cT, T_j^1\}$
\ENDWHILE
\STATE $\cT \leftarrow \{\cT, T_L^1\}$
\end{algorithmic}
\end{algorithm}

In multiple packets case, since we have no initial knowledge about which nodes should decode which packets and in which order, we believe it is difficult to design an effective decoding order initialization algorithm. We therefore initialize the decoding order based purely on the one with single transmitted packet. Suppose in single packet case the following decoding order is found  
 
\begin{align*}
 \cT =\bigg \{ T_1^1, T_i^1, \ldots, T_L^1 \bigg \},
\end{align*}
the decoding order for multi-packet case is simply chosen to be
\begin{align*}
\cT =\bigg \{ T_1^1, \ldots, T_1^N, T_i^1, \ldots, T_i^N, \ldots, T_L^N \bigg \}.
\end{align*}

 
 \subsubsection{Decoding order updates}
Based on the solution to the LP, the algorithm first searches for all those $m$'s such that $\Delta_m = 0$. Then for those $m$, it checks for the situation where $T_i^c$ is followed by $T_L^c$ for some $i \neq L$. If it finds one, it drops the corresponding $T_i^c$ and reruns the LP. If it does not find one, it swaps the order of $T_{m-1}$ and $T_m$. If the swapped decoding events are all the same as the ones in the previous iteration, and if there was no decoding event dropped in the previous iteration, the algorithm terminates. We present the sketch of the algorithm in Algorithm~\ref{alg.algorithm}. 
\begin{algorithm}
\caption{The routing and resource allocation algorithm}
\label{alg.algorithm}
\begin{algorithmic}[1]
\STATE Choose a LP $\cP$
\STATE Initialize a decoding order $\cT$
\STATE DoAgain $\leftarrow$ \TRUE
	\WHILE {DoAgain}
	\STATE result $\leftarrow$ LP($\cT$, $\cP$)
	\STATE index $\leftarrow$ SearchForZeroDeltas(result)
	\STATE ToDrop $\leftarrow$ SearchForNodeDecodeAfterDestination(index)
	\IF {ToDrop $\neq 0$}
		\STATE DoAgain $\leftarrow$ \TRUE
		\STATE drop the decoding event
	\ELSIF {drop $=0$ \textbf{and} index $=$ index in last iteration \textbf{and} no event is dropped in last iteration}
		\STATE DoAgain $\leftarrow$ \FALSE
		\STATE \textbf{break}
	\ELSE
		\STATE DoAgain $\leftarrow$ \TRUE
		\STATE SwapEvents(result)
	\ENDIF
\ENDWHILE 
\end{algorithmic}
\end{algorithm}

Because of the exponential number of orderings we expect the problem of finding the optimal decoding order to be NP-hard. The sub-optimality of our heuristic algorithm comes from the fact that the excluded decoding events may actually be helpful. Without a mechanism to ``re-introduce'' the excluded events our algorithm is not expected to achieve global minimum for all networks.   

\subsubsection{Characteristics of final route}
Our algorithm has the property that different packets may take different paths toward the destination under the per-node resource constraint. This exploits the spatial route diversity. For example, the algorithm may schedule the first $3$ packets to pass through a set of relay nodes and schedule the next $3$ packets to pass through yet another set of relay nodes. The intuition is that instead of choking an existing route and overloading the already busy nodes, the algorithm balances the traffic by invoking other free nodes to fully utilize the network. 

\subsubsection{Challenges}
The challenges we face in designing this multi-packets routing and resource allocation algorithm are more difficult than the ones the authors encountered in \cite{Draper_2008}. Essentially, we are trying to sort out a correct order for \textit{both} nodes and packets. Our algorithm may be stuck in local solutions without the ability to escape. However, as we show in simulation, for smaller networks, although it may not always find the optimal solution, the local solution it achieves is still promising.  

\section{Numerical Results}
\label{sec.simulation_results}
In this section we present detailed simulation results for the algorithm. These results exemplify the basic properties of using mutual information accumulation in cooperative communication.
\subsection{A simple network}

To better illustrate how our algorithm exploits the spatial route diversity, we first consider a simple four-node example in which a data file of size $20$ is equally packetized into two packets of size $10$. Consider the diamond network in Fig.~\ref{fig.diamond_example}. Links with nonzero spectral efficiency (bits/sec/Hz) are shown. Each node is allocated with $1$ unit of bandwidth and the source (labeled as $1$) transmits these two packets to the destination (labeled as $4$).  

\begin{figure}[!t] 
  \centerline{\epsfig{figure=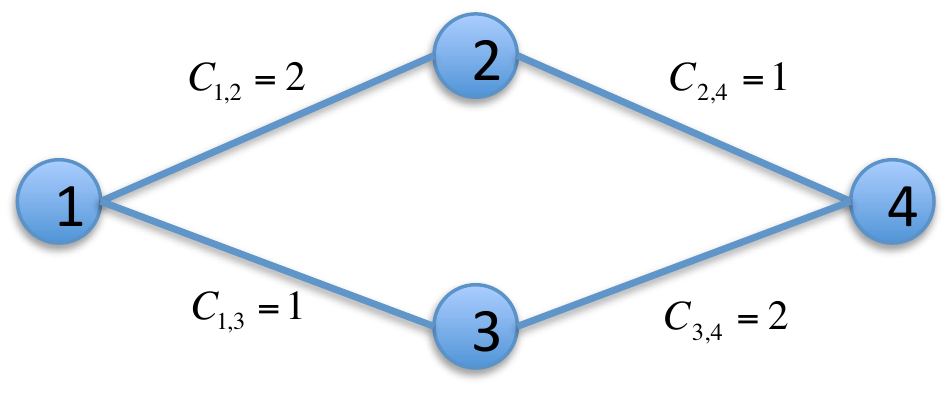,width=5cm}}
  \caption{A diamond network example. Note that transmitting a single packet on route $1-3-4$ or $1-2-4$ takes the same amount of time.}
  \label{fig.diamond_example}
\end{figure} 

Running our algorithm returns the following routing and resource allocation scheme. Starting from time $0$, node $1$ broadcasts packet $1$ for $5$ units of time. By the end of time $5$, node $2$ decodes packet $1$ and node $3$ accumulates half of packet $1$. Starting from time $5$, node $2$ broadcasts packet $1$ to node $4$ while node $1$ broadcasts packet $2$. By the end of time $15$, node $4$ decodes packet $1$ and node $3$ decodes packet $2$. Note that by the end of time $10$, node $2$ decodes packet $2$. Then starting from time $15$, node $2$ and $3$ transmit packet $2$ to node $4$ at the same time and node $4$ decodes packet $2$ after $3.33$ units of time. The total transmission time is thus $5+10+3.33 = 18.33$ units of time. In this example, node $3$ never decode packet $1$ thus packet $1$ does not pass through route $1-3-4$. Whereas node $3$ decodes packet $2$ and it cooperates with node $2$ to deliver packet $2$ to destination. 

It is interesting to see that if two packets were routed through the same path, the total transmission time would be $2 \times 15 = 30$. If the data file itself is routed without packetizing, the total transmission time would be $30$ no matter which route ($1-3-4$ or $1-2-4$) it takes. This suggests that possible improvement can be achieved if we divide data files into smaller packets. We further illustrate this property in the following section.

\subsection{Simulation on general networks}
We now study the performance of our algorithm on more general networks. For better illustration, we choose bandwidth as our single resource element being allocated.  

We simulate the performance of our algorithm on $10$ randomly scattered wireless nodes within a unit square. The source, denoted as $1$, is always located at $[0.2, 0.2]$ and the destination, denoted as $L$, is always located at $[0.8, 0.8]$. The remaining $8$ nodes are randomly placed according to the uniform distribution in the unit square. The channel gain $h_{i, j}$ is deterministically related to the Euclidean distance $d_{i, j}$ between node $i$ and $j$ as $h_{i, j} = d_{i, j}^{-2}$. The spectral efficiency between node $i$ and $j$ is given by ({\ref{eq.spectral_efficiency}). 

\subsubsection{Per-node bandwidth constraint}
We first consider the system with the per-node bandwidth constraint. We consider a system with single arrival of a data file of size $20$ which is equally packetized into $N$ packets. Each packet is of size $20/N$. We also let $P_i = P = 1$, $W_i = 1$ and $N_0 / 2 = 1$. We simulate networks to estimate the cumulative distribution function (CDF) of the total transmission time when we either do not subdivide the data file (i.e., the same problem as considered in \cite{Draper_2008}) or divide it into $2$ or $3$ packets for transmission. To further illustrate the benefits of using mutual information accumulation, we also plot the transmission performance on the route obtained using Dijkstra's algorithm \cite{dijkstra} with and without mutual information accumulation in Fig.~\ref{fig.pernode_bw_sim}.

\begin{figure}[!t] 
  \centerline{\epsfig{figure=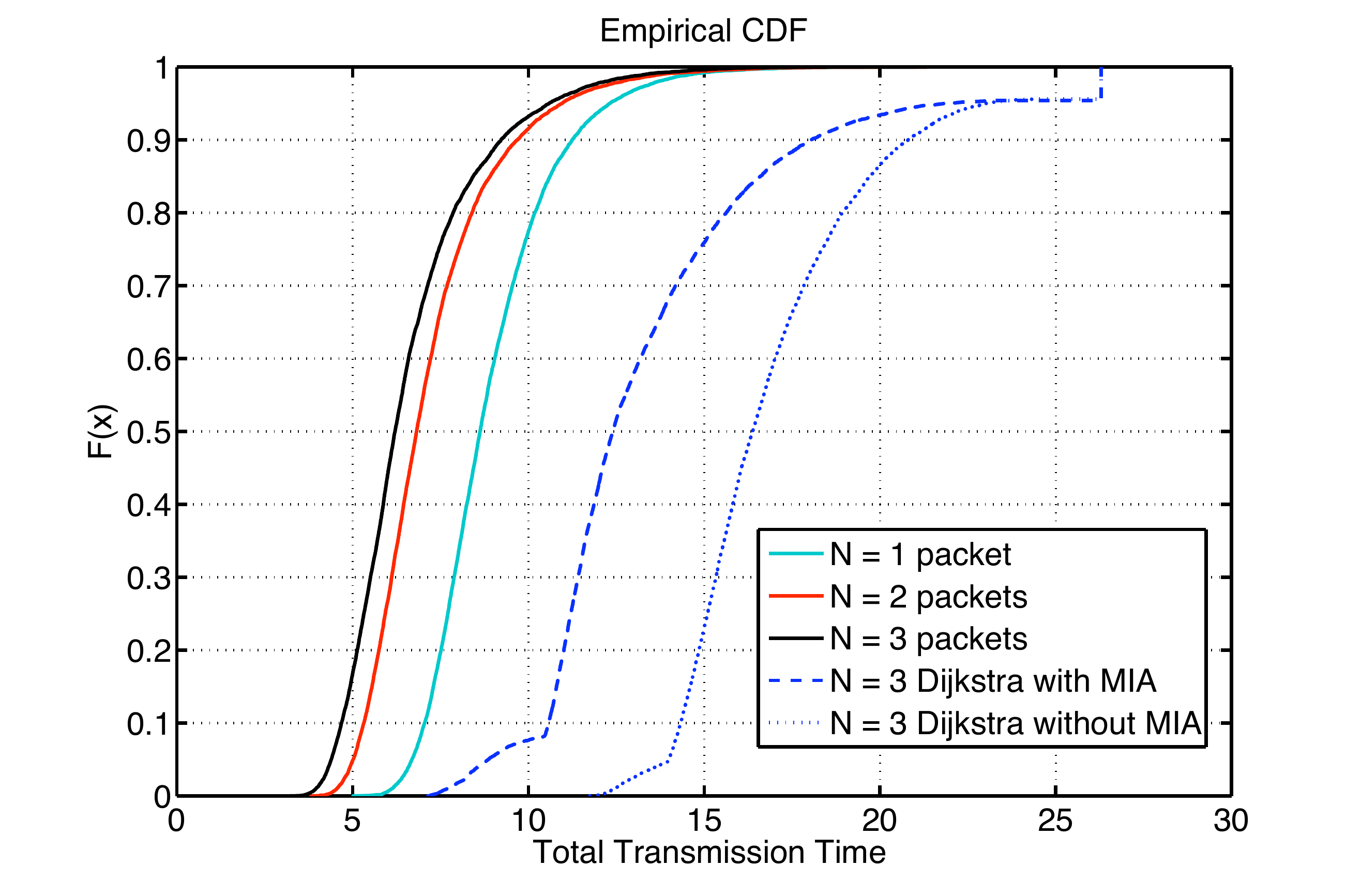,width=9cm}}
  \caption{Cumulative distribution of total transmission time under N = 1, 2, and 3. While networks randomly generated within the unit square have constant number of uniformly distributed nodes, the Dijkstra routes on these networks may contain different number of nodes. This contributes to the non-smoothness in the Dijkstra algorithm simulation plot.}
  \label{fig.pernode_bw_sim}
\end{figure} 

We see that dividing a data file into packets improves the performance by reducing the total transmission time. If we keep dividing it into more and more packets, we expect that the performance eventually converge as if the data file is treated as fluid. Comparing with Dijkstra's algorithm, we also notice that the decrease in total transmission time comes from two causes; the use of mutual information accumulation and the ability to exploit other available routes and resource within a network, not merely confined in a pre-determined Dijkstra's route. This illustrates how our algorithm exploits the spatial route diversity. 

\subsubsection{Per-node bandwidth constraint with arrival process}
We now consider a deterministic arrival process under the per-node bandwidth constraint. Consider a randomly generated network shown in Fig.~\ref{fig.arrprocs_ex}. We now route $K=3$ data files with inter-arrival time between two consecutive files, each is of size $20$, from the source to the destination. We also let $P_i = P = 1$, $W_i = 1$ and $N_0 / 2 = 1$. Instead of using total transmission time as our objective, we try to minimize the average transmission time given in (\ref{eq.avg_tx_time}). We simulate the average transmission time required when inter-arrival time ($\tau_2 - \tau_1$ or $\tau_3 - \tau_2$) varies from $0$ to $15$. 
\begin{figure}[!t] 
  \centerline{\epsfig{figure=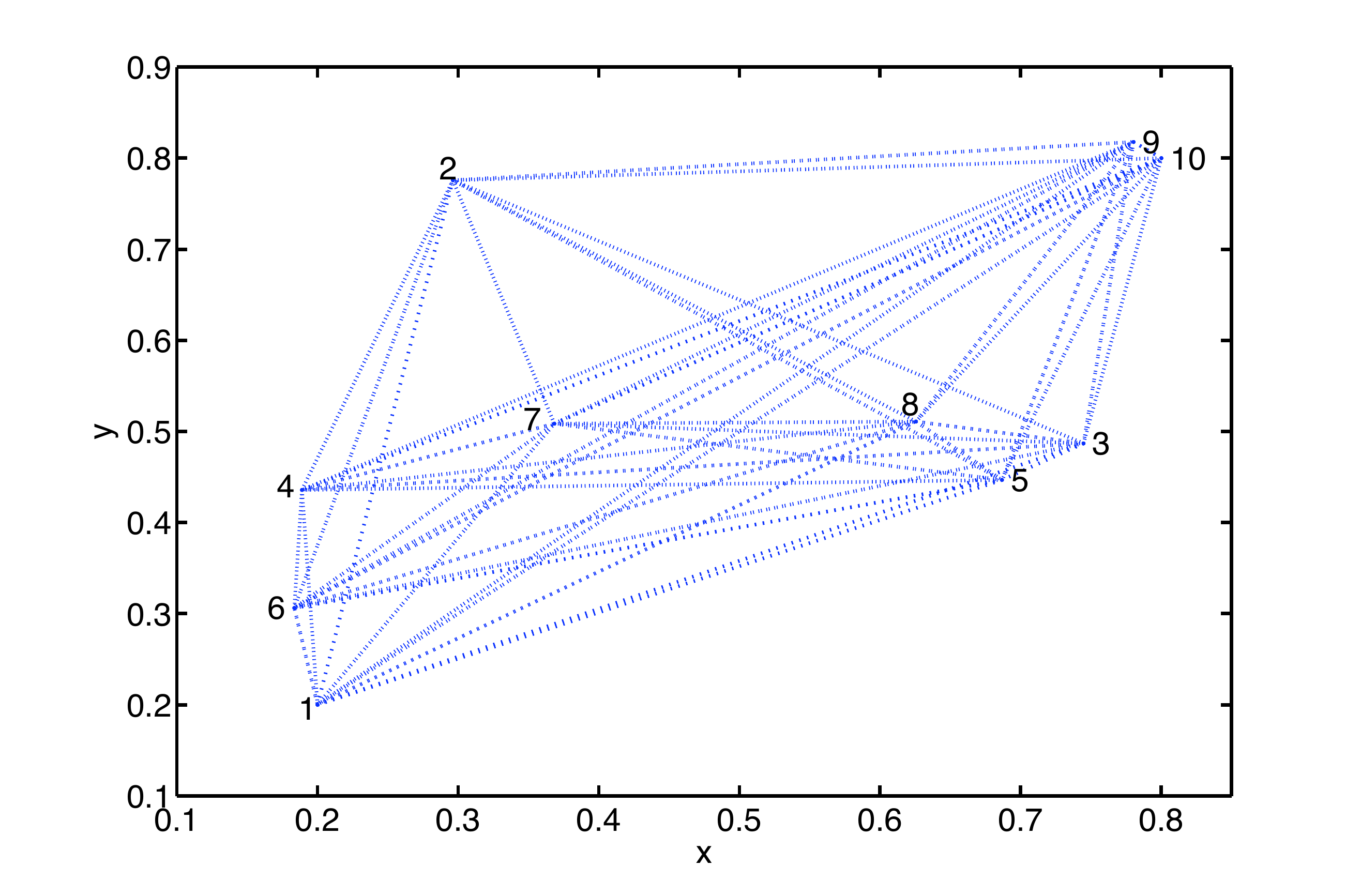,width=9cm}}
  \caption{A randomly generated network. Transmitting a single data file of size $20$ on this network takes $7.1576$ units of time.}
  \label{fig.arrprocs_ex}
\end{figure} 
 
The simulation shows the exploitation of the temporal route diversity. For this given network, transmitting a single data file of size $20$ requires $7.1576$ units of time. If the inter-arrival time is larger than this number, then transmitting $3$ arriving data files means repeating the single file routing and resource allocation policy three times. Thus the average transmission time is the same as transmission time for a single data file. When the inter-arrival time is small, files compete with each other for resource, thus the average transmission time increases as shown in Fig.~\ref{fig.pernode_arrprcs}

\begin{figure}[!t] 
  \centerline{\epsfig{figure=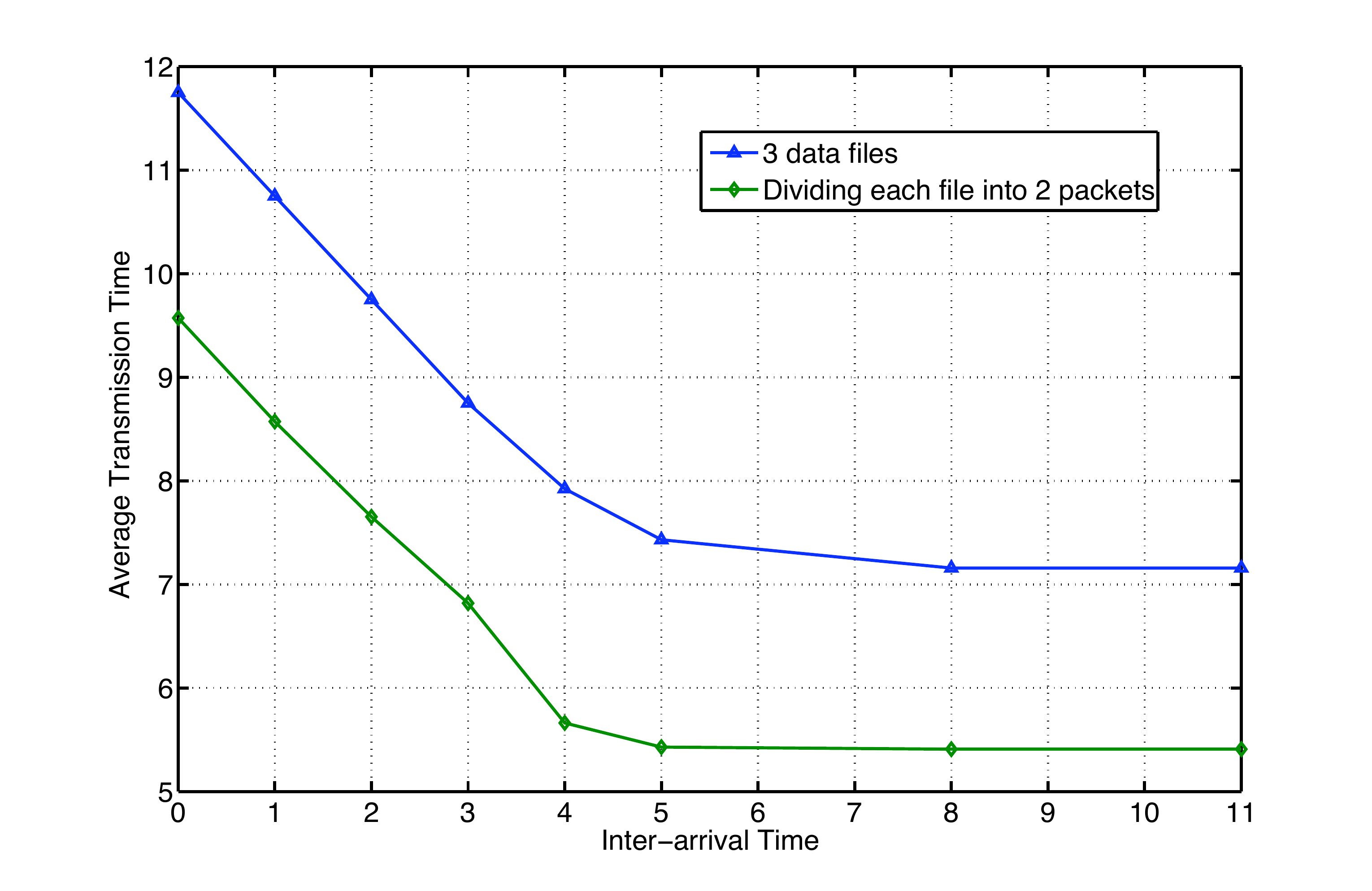,width=9cm}}
  \caption{Average transmission time under different inter-arrival times.}
  \label{fig.pernode_arrprcs}
\end{figure}  

We now divide each data file equally into $2$ packets of size $10$ (thus $N = 2K = 6$). Since under the per-node resource constraint dividing a big data file into packets better utilizes network resource, we have a performance gain shown in Fig.~\ref{fig.pernode_arrprcs}. Note that the inter-arrival time is between every two packets of size $10$. 
 
\subsubsection{Sum bandwidth constraint}
We next simulate the performance of the system with the sum bandwidth constraint. We consider a single data file of size $20$ which is equally divided into $N$ packets. We let $W_T = 10$, $P_i = P = 1$, $N_0/2 = 1$. Note that the total bandwidth is the sum of individual bandwidth $W_i$ in the per-node bandwidth constraint simulation setting. 

The simulation result is shown in Fig.~\ref{fig.sum_bw_sim}. Interestingly, we note that under the sum bandwidth constraint, dividing a data file into more and more packets would not help to decrease the total transmission time. This is the result of Theorem~\ref{tm.sum_bw_theorem2}; since each packet is also $\frac{1}{N}$ of the size of the data file, the total transmission time would not change. 

We also notice that in Fig.~\ref{fig.sum_bw_sim}, the performance under the sum bandwidth constraint is better than that under the per-node bandwidth constraint. The reason is that in allocating each node a fixed amount of bandwidth as in the per-node constraint, the nodes which cannot transmit waste the resource assigned to them. In the sum bandwidth constraint, the network gathers all available resource and distributes it to serve the most needed nodes; the nodes which cannot transmit are not assigned with resource.   

\begin{figure}[!t] 
  \centerline{\epsfig{figure=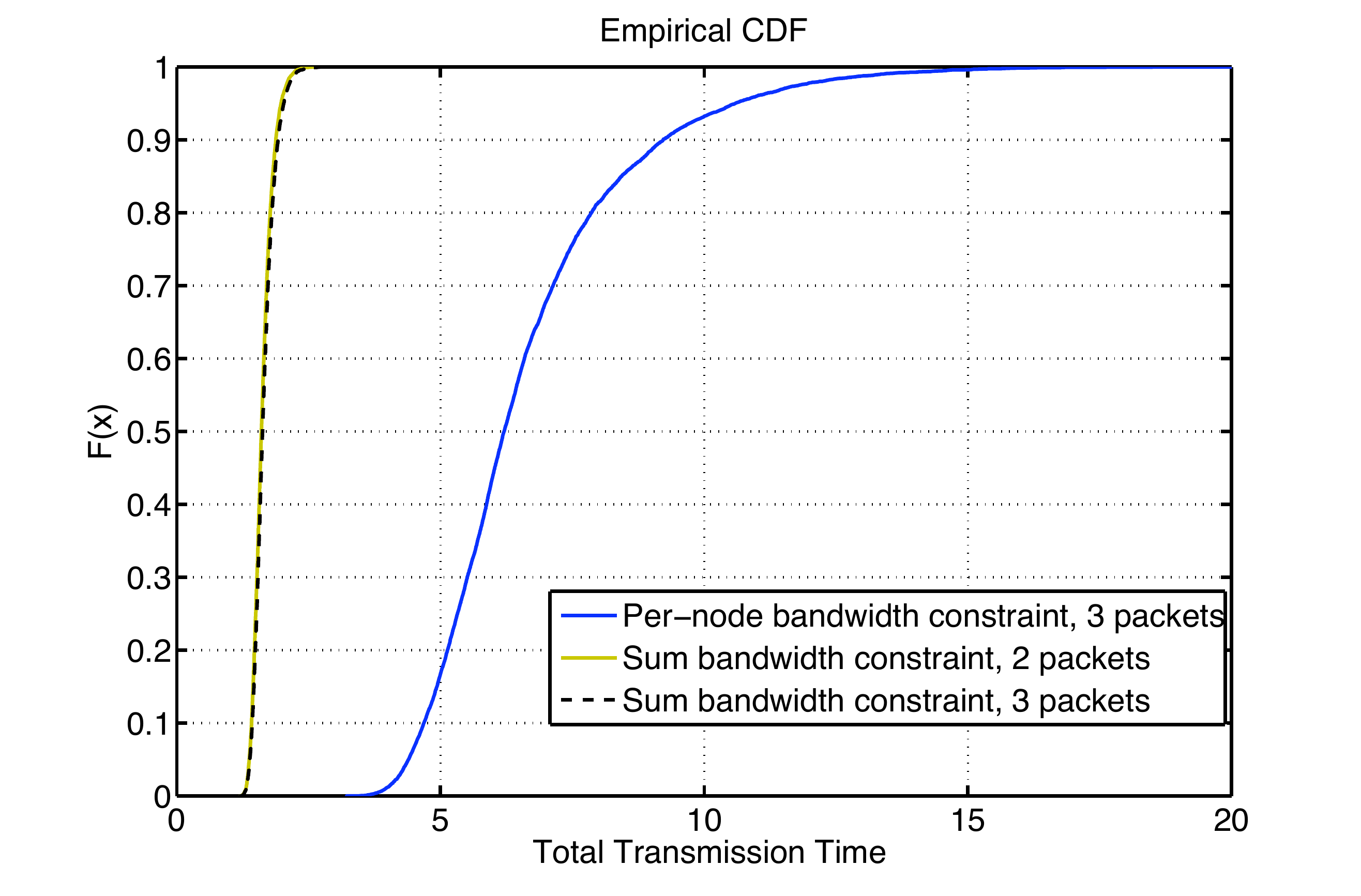,width=9cm}}
  \caption{Simulation of the sum bandwidth constraint and comparison with the per-node bandwidth constraint. }
  \label{fig.sum_bw_sim}
\end{figure} 


\section{Summary and Future Work}
\label{sec.discussion_and_conclusion}
In this paper we study the problem of cooperative communication using mutual information accumulation. We design and prototype a routing and resource allocation algorithm based on solving a LP based problem iterating between two subproblems; finding the best decoding order and finding the best resource allocation given a decoding order. Our algorithm exploits both the spatial and temporal route diversity and suggests that performance can be improved by dividing a big data file into several small packets. Our algorithm also supports packets with different arrival times. 

In some scenarios, it is not possible to centralize the routing and resource allocation scheme. Limitations on centralized algorithm are constraining in large scale networks where aggregating the channel state information can incur unacceptable overheads. The centralized algorithm is also not desired in temporally varying networks where central controller often has trouble in updating the estimates of channel state information promptly. Developing a distributed algorithm would be an important direction for future research.  

Conventional wireless nodes have the ability to adjust the transmitting power. With individual transmitting power as a varying variable, it is therefore natural to consider power allocation optimization. We note that individual transmitting power affect the spectral efficiency according to ({\ref{eq.spectral_efficiency}). Thus maintaining both the time-bandwidth product and the transmitting power as optimization variables makes the decoding constraint nonlinear. Understanding how power allocation would affect the transmission performance is an interesting topic for investigation. 

Consider a cooperative broadcast problem in which the objective of the network is to deliver packets from the source to all other nodes. Since every node is a destination, we do not need to drop a decoding event. This preserves some optimality of our algorithm. However, it is still unknown that if our algorithm is the most effective one. It would be interesting to investigate possible extensions of our algorithm tailored for this problem.

\bibliographystyle{IEEEtran}
\bibliography{myrefs}
\end{document}